\documentclass[a4paper,UKenglish,thm-restate]{article}
\pdfoutput=1

\usepackage{amssymb,amsthm,amsmath,authblk,geometry,xcolor,mathtools}
\usepackage[ocgcolorlinks,colorlinks=true,citecolor=blue,linkcolor=black,urlcolor=blue]{hyperref}

\title{Intersection Patterns of Set Systems on Manifolds with Slowly Growing Homological Shatter Functions}

\author{
  Sergey Avvakumov\\
  School of Mathematical Sciences\\
  Tel Aviv University\\
  Tel Aviv 69978, Israel\\
  \and 
  Marguerite Bin\\
  Universit\'e de Lorraine, CNRS, INRIA\\
  LORIA,  F-54000 Nancy, France\\
  \and
  Xavier Goaoc\\
  Universit\'e de Lorraine, CNRS, INRIA\\
  LORIA,  F-54000 Nancy, France\\
  }

\newcommand{\R}{\mathbb{R}}
\newcommand{\Z}{\mathbb{Z}}
\newcommand{\N}{\mathbb{N}}
\newcommand{\s}{\mathbb{S}}

\newcommand{\p}{\mathcal{P}}
\newcommand{\C}{\mathcal{C}}
\newcommand{\F}{\mathcal{F}}
\newcommand{\G}{\mathcal{G}}
\newcommand{\M}{\mathcal{M}}

\DeclareMathOperator{\h}{h}
\DeclareMathOperator{\conv}{conv}
\DeclareMathOperator{\rad}{r}
\DeclareMathOperator{\ch}{ch}
\DeclareMathOperator{\fh}{fh}
\DeclareMathOperator{\supp}{supp}
\DeclareMathOperator{\Csing}{C_*^{sing}}
\DeclareMathOperator{\hc}{hc}
\DeclareMathOperator{\im}{im}
\DeclareMathOperator{\sd}{sd}

\newcommand{\card}[1]{\textrm{Card}\pth{#1}}
\newcommand{\defn}[1]{\emph{\color{violet}#1}}
\newcommand{\set}[1]{\left\{#1\right\}}
\newcommand{\pth}[1]{\left(#1\right)}

\newtheorem{replemma}{Lemma}
\newtheorem{theorem}{Theorem}
\newtheorem{conjecture}[theorem]{Conjecture}
\newtheorem{corollary}[theorem]{Corollary}
\newtheorem{lemma}{Lemma}[section]
\newtheorem{proposition}[lemma]{Proposition}

\begin{document}

\maketitle

\begin{abstract}
  A theorem of Matou\v sek asserts that for any $k \ge 2$, any set system whose shatter function is $o(n^k)$ enjoys a fractional Helly theorem of order $k$: in the $k$-wise intersection hypergraph, positive density implies a linear-size clique. Kalai and Meshulam conjectured a generalization of that phenomenon to \emph{homological shatter functions}. It was verified for set systems with bounded homological shatter functions and whose ground set has a \emph{forbidden homological minor} (which includes $\R^d$ by a homological analogue of the van Kampen-Flores theorem). We present two contributions to this line of research:

  \begin{itemize}
  \item We study homological minors in certain manifolds (possibly with boundary), for which we prove analogues of the van Kampen-Flores theorem and of the Hanani-Tutte theorem.\medskip
  \item We introduce graded analogues of the Radon and Helly numbers of set systems and relate their growth rate to the original parameters. This allows to extend the verification of the Kalai-Meshulam conjecture to sufficiently slowly growing homological shatter functions.
  \end{itemize}
\end{abstract}

\section{Introduction}

A classical line of research in discrete geometry investigates generalizations of properties of convex sets beyond convexity, with a particular attention to topological conditions. At least three distinct lines of inquiry emerged:

\begin{description}
\item[(A)] Generalizing families of convex sets into \emph{acyclic/good covers}, meaning set systems in topological spaces such that every subfamily has empty or homologically/homotopically trivial intersection; an early example is Helly's topological theorem~\cite{helly1930systeme}, see also the survey of Tancer~\cite{tancer2012intersection} for an overview.\medskip

\item[(B)] Reformulating properties of convex sets of $\R^d$ as properties of linear maps into $\R^d$ and investigating their generalizations to \emph{continuous} maps into $\R^d$, typically via theorems of Borsuk-Ulam type; an early example is the topological Radon theorem of Bajmóczy and Bárány~\cite{bajmoczy1979common} and more examples can be found \emph{e.g.} in the survey of Bárány and Soberón~\cite{barany2018tverberg}.\medskip

\item[(C)] Analyzing set systems whose nerve enjoy properties of nerves of convex sets, like \emph{$d$-collapsibility} or \emph{$d$-Lerayness}; an early example is the sharpening of the Fractional Helly theorem by Kalai~\cite{kalai1984intersection} and the work of Alon et al.~\cite{transversal-hypergraph} establishes several landmark results.
\end{description}

\noindent
Two decades ago, Kalai and Meshulam~\cite{KM} proposed conjectures relating approaches~(A) and~(C), towards what they called a theory of \emph{homological VC dimension}. First, in order to generalize the notion of \emph{good cover}, let us measure the complexity of the intersection patterns of a set system $\F$ in a topological space by its \defn{$h$th homological shatter function}

\begin{equation}\label{eq:homshat}
  \phi_{\F}^{(h)}\colon\left\{\begin{array}{rcl}
        \mathbb{N} & \rightarrow & \mathbb{N} \cup \{\infty\} \\
         k& \mapsto & \sup \left\{\left.\tilde{\beta_i}\pth{\bigcap\limits_{F\in \mathcal{G}} F;\Z_2} \ \right| \ \G\subset \F, |\G|\leq k, 0\leq i\leq h\right\}.
    \end{array}\right.
\end{equation}

\noindent
Here $h$ is some fixed parameter, and $\tilde\beta_i(\cdot;\Z_2)$ is the $i$th reduced Betti number with coefficients in $\Z_2$. (The set systems $\F$ with $\phi_{\F}^{(\infty)} \equiv 0$ are the acyclic covers and include good covers and convex sets.) The  conjectures are about nerves of set systems whose intersection patterns have polynomially-growing topological complexity. In particular, their combination~\cite[Conjectures~6~and~7]{KM} (see also~\cite[Conjecture~1.9]{steppingUp}) implies that polynomially-growing homological shatter functions give rise to a ``\emph{positive density implies big clique}'' phenomenon:

\begin{conjecture}[Kalai and Meshulam]
    \label{conj:kalai-meshulam}
    For any $d \in \N$ and any function $\Psi: \N \to \N$ such that $\Psi(n) = O(n^d)$, there exists $\beta: (0,1) \to (0,1)$ such that the following holds. For any $\alpha>0$ and any set system $\F$ in $\R^d$ with $\phi_{\F}^{(d)} \le \Psi$, if a proportion $\alpha$ of the $(d+1)$-element subsets of $\F$ have nonempty intersection, then some $\beta(\alpha)|\F|$ members of $\F$ have a point in common.
\end{conjecture}

\noindent
Such conditions that a positive density of $d$-faces in the nerve implies the existence of a linear-size face is called a \emph{fractional Helly theorem} (see Section~\ref{s:cbk}).

\medskip

Conjecture~~\ref{conj:kalai-meshulam} is a topological analogue of a theorem of Matou\v sek~\cite{matousek2004bounded} which asserts that every set system with polynomial (combinatorial) shatter function enjoys a fractional Helly theorem. Families of semi-algebraic sets of bounded complexity in $\R^d$ are an interesting example at the intersection of these two realms: classical results on the Betti numbers of algebraic varieties (see~\emph{e.g.}~\cite{basu2018multi} and the references therein) ensure that both their combinatorial shatter functions and their homological shatter functions are polynomial.

\medskip

Conjecture~\ref{conj:kalai-meshulam} was confirmed for functions $\Psi$ that are bounded (\cite[Corollary~1.3]{steppingUp}, building on~\cite{Holmsen2020,Holmsen2021,patakova2022bounding}). The main contribution of the present paper is to extend that confirmation to some diverging homological shatter functions and to set systems on certain manifolds.

\subsection{Context and motivation}\label{s:context}

Before we state our results precisely (in Section~\ref{s:results}) let us provide some context and motivation, as well as introduce some necessary terminology.

\subsubsection{Combinatorial background: convexity parameters of set systems}\label{s:cbk}

The classical theorems of Helly, Radon and Carathéodory have initiated a rich theory of the combinatorial properties of convexity, whose landmarks include the centerpoint theorem, Tverberg's theorem, the colorful Helly and Carathéodory theorems, the fractional Helly theorem, the selection lemma, the weak $\varepsilon$-net theorem, the $(p,q)$-theorem, etc. We refer the interested reader to the monograph of Bárány~\cite{barany2021combinatorial} and the textbook of Matou\v{s}ek~\cite{matousek2013lectures}. These classical convexity theorems have algorithmic consequences for instance in optimization and geometric data analysis~\cite[$\mathsection~6-7$]{de2019discrete} or in property testing~\cite{chakraborty2018helly}, and one motivation for their extension beyond the convex setting is that several of these benefits generalize as well~\cite{msw-sblp-92,epstein_et_al:LIPIcs.ICALP.2020.98}.

\medskip

We can associate to any set system $\F$ with ground set $X$ some parameters inspired by convexity properties, for instance:

\begin{itemize}
\item The \defn{Helly number} $\h_\F$ of $\F$ is the smallest integer $h$ with the following property: If in a finite subfamily $\G \subset\F$, every $h$ members of $\G$ intersect, then $\G$ has nonempty intersection. If no such $h$ exists, we set $\h_\F=\infty$.\medskip

\item The \defn{Radon number} $\rad_\F$ of $\F$ is the smallest integer $r$ such that every $r$-element subset $S\subset X$ can be partitioned into two nonempty parts $S=P_1\sqcup P_2$ such that $\conv_{\F}(P_1) \cap \conv_{\F}(P_2) \neq \emptyset$. (The set $\conv_{\F}(P)$, the \emph{$\F$-convex hull} of a subset $P\subset X$, is the intersection of all the members of $\F$ that contain $P$.) If no such $r$ exists, we set $\rad_\F=\infty$.
\end{itemize}

\noindent
Hence, letting $\C_d$ denote the set of all halfspaces in $\R^d$, Radon's lemma asserts that $\rad_{\C_d} = d+2$ and Helly theorem that $\h_{\C_d} = d+1$. A classical result by Levi~\cite{Levi_radon_implies_helly} asserts that for every set system $\F$ we have $\h_{\F} \leq \rad_\F-1$. (This is often stated for convexity spaces but it holds for set systems, see Appendix~\ref{app:convspaces}.) Similar relations between such parameters have been investigated over the years, like for instance the partition conjecture of Eckhoff~\cite{ECKHOFF200061} refuted by Bukh~\cite{bukh2010radon}.

\medskip

Conjecture~\ref{conj:kalai-meshulam} pertains to a parameter inspired by the \emph{fractional Helly theorem}~\cite{katchalski1979problem,kalai1984intersection}, which asserts that in $(d+1)$-wise intersection hypergraphs of convex sets of $\R^d$, positive density implies a linear-size clique. Here is the associated parameter:

\begin{itemize}
\item The \defn{fractional Helly number} $\fh_\F$ of $\F$ is the smallest integer $s$ such that there exists a function $\beta_\F:(0,1)\rightarrow (0,1)$ with the following property: For every finite subfamily $\F'\subset \F$, whenever a fraction $\alpha$ of the $s$-tuples of $\F'$ have nonempty intersection, a subset $\G$ of $\F'$ of size $\beta_\F(\alpha)|\F'|$ has nonempty intersection.
\end{itemize}

\noindent
Again, the fractional Helly theorem states that $\fh_{\C_ d} = d+1$. The significance of the fractional Helly number was highlighted by Alon et al.~\cite{transversal-hypergraph}, who proved that intersection-closed set systems with bounded fractional Helly number enjoy a \emph{weak $\varepsilon$-net theorem}, a \emph{Tverberg-type theorem}, a \emph{selection lemma}, etc. It is tempting to reformulate Conjecture~\ref{conj:kalai-meshulam} as

\begin{quote}
  \emph{For any function $\Psi: \N \to \N$ such that $\Psi(n) = O(n^d)$, every set set system $\F$ in $\R^d$ such that $\phi_{\F}^{(d)} \le \Psi$ has fractional Helly number at most $d+1$.}
\end{quote}

\noindent We note, however, that this statement is weaker than Conjecture~\ref{conj:kalai-meshulam} in that it does not assert that the function $\beta()$ underpinning the fractional Helly number depends only on $\Psi$ and $d$.

\subsubsection{Homological background: homological minors}

A classical way to extend the theory of planar graphs is to consider embeddings of graphs into surfaces and of simplicial complexes into $\R^d$ or other topological spaces.

\medskip

There is a rich theory of embedding of graphs on surfaces, both structural (a classic being the Heawood inequality~\cite{ringel2012map}) and computational (\emph{e.g.} the use of graph genus for parameterized complexity). Let us mention, in particular, the strong Hanani-Tutte theorem which asserts that a graph is planar if it can be drawn so that every pair of independent edges cross an even number of times. This statement generalizes to the projective plane~\cite{pelsmajer2009strong,de2017direct} but was found to fail in genus~4~\cite{fulek2019counterexample}. 

\medskip

Going to dimension higher than $2$ changes the nature of the problems drastically, already because Fáry's theorem no longer holds. (For every $d \ge 3$ there are simplicial complexes that embed in $\R^d$ piecewise linearly but not linearly.) It is thus sometimes convenient to relax the notion of embedding and work with chain maps, and this was done in particular to analyze intersection patterns~\cite{hellyBetti,patakova2022bounding,steppingUp} using a notion of \emph{homological minors}~\cite{wagner2012minors}.

\medskip

Formally, the \defn{support of a singular chain} is the union of (the images of) the singular simplices with nonzero coefficient in that chain, and the  \defn{support of a simplicial chain} is the subcomplex induced by the simplices with nonzero coefficient in that chain.  We write $\supp(\sigma)$ for the support of a (singular or simplicial) chain $\sigma$. A chain map $a:C_*(K)\to \Csing(X)$ (resp. $a:C_*(K)\to \C_*(T)$) is \defn{nontrivial} if, for every vertex $v$ of $K$, the support of $a(v)$ has odd size. Two faces in a simplicial complex $K$ are \defn{adjacent} if they have at least one vertex in common.  A \defn{homological almost-embedding} of a simplicial complex $K$ into a topological space~$X$ (resp. into another simplicial complex $L$) is a nontrivial chain map $a:C_*(K)\to \Csing(X)$ (resp. $a:C_*(K)\to C_*(L)$) such that any two non-adjacent faces $\sigma,\tau \in K$ have images with disjoint support, that is $\supp(a(\sigma)) \cap \supp(a(\tau)) = \emptyset$. In particular, for every embedding $f$ the associated chain map $f_\#$ is a homological (almost) embedding.

\medskip

Let $\Delta_N$ denote the $N$-dimensional simplex and for $K$ a simplicial complex let $K^{(t)}$ denote its $t$-dimensional skeleton. It turns out that there is a homological version of the van Kampen-Flores theorem (see for instance~\cite[Corollary~14]{hellyBetti}):

\begin{theorem}
  \label{t:homological-vKF}
  For any $d \ge 1$, $\Delta_{d+2}^{(\lceil d/2 \rceil)}$ does not homologically almost embed into $\R^d$.
\end{theorem}

\noindent
A simplicial complex $K$ is a \defn{homological minor} of a topological space or simplicial complex~$X$ if there is a homological almost-embedding of $K$ into $X$. Theorem~\ref{t:homological-vKF} thus asserts that $\Delta_{d+2}^{(\lceil d/2 \rceil)}$ is not a homological minor of $\R^{d}$, \emph{i.e.}, that $\R^{d}$ has $\Delta_{d+2}^{(\lceil d/2 \rceil)}$ as \defn{forbidden homological minor}.

\subsubsection{Parameters of set systems with a forbidden homological minor}

Matou\v sek~\cite{m-httucs-97} bounded the Helly number of topological set systems in $\R^d$ in which every subfamily intersects in a bounded number of connected components, all contractible. His approach starts from a set systems in $\R^d$, uses Ramsey theory to build a map from $\Delta_{d+2}^{(\lceil d/2 \rceil)}$  into $\R^d$ that is ``constrained'' by the known intersection patterns of $\F$ so that the intersection forced by the van Kampen-Flores theorem reveals a new intersection in~$\F$.

\medskip

This approach was generalized by Goaoc et al.~\cite{hellyBetti} to set systems in $\R^d$ of bounded $(\lceil d/2 \rceil)$-level topological complexity, where the \defn{$h$-level topological complexity} of $\F$ is the maximum over $\N$ of the homological shatter function $\phi_{\F}^{(h)}$, that is

\begin{equation}\label{eq:topcomp}
  \hc_\F^{(h)} := \max\left\{\max_{0 \le i < h}\left.\tilde{\beta_i}\pth{\bigcap\limits_{A\in \G} A;\Z_2} \ \right| \ \G\subset \F\right\}.
\end{equation}

\noindent
That generalization relied on homological minors and replaced the construction of the ``constrained map'' by the (simpler) construction of a ``constrained chain map''. The only aspect of the method that is specific to $\R^d$ is the use of the forbidden homological minor given by Theorem~\ref{t:homological-vKF}, so the approach readily generalizes to set systems whose ground set is a topological space with a forbidden homological minor, see the discussions in~\cite[$\mathsection 5.2$]{patakova2022bounding} and~\cite[$\mathsection 2.2$]{steppingUp}. This method was refined and extended to analyze other parameters of set systems whose ground set has a forbidden homological minor~\cite{patakova2022bounding,patak2025sharper,steppingUp}.

\subsubsection{Proof of Conjecture~\ref{conj:kalai-meshulam} for bounded homological shatter functions}

Conjecture~\ref{conj:kalai-meshulam} was proven for bounded homological shatter functions in three steps: 

\begin{itemize}
\item First, Holmsen and Lee~\cite[Theorem~1.1]{Holmsen2021} proved that the fractional Helly number of \emph{any} set system can be bounded by a function of its Radon number. Specifically, letting $\Psi_{\rad \to \fh}(x)$ denote the supremum of the fractional Helly number of a set system with Radon number $x$, they proved  that for every $r \ge 3$ we have $\Psi_{\rad \to \fh}(r) \le r^{r^{\lceil \log_2r \rceil}+r\lceil \log_2r \rceil}$.\medskip

\item Second, Patáková~\cite[Theorem~2.1]{patakova2022bounding} proved that the Radon number of \emph{any} set system whose ground set has $K$ as forbidden homological minor can be bounded by a function of its $(\dim K)$-level topological complexity.\medskip

\item Third, Goaoc, Holmsen and Patáková~\cite[Theorem~1.2]{steppingUp} proved that for every set systems $\F$ whose ground set has $K$ as forbidden homological minor, if the fractional Helly number $\fh_\F$ is bounded then it is at most $\mu(K)+1$, where $\mu(K)$ denotes the maximum sum of dimensions of two disjoint simplices in $K$.
\end{itemize}

\subsection{Statement of the results}\label{s:results}

Our main contribution is to confirm Conjecture~\ref{conj:kalai-meshulam} for some diverging homological shatter functions and when the ground set is a manifold. This decomposes into five independent results.

\medskip

Throughout the paper, we work with compact piecewise-linear (PL) manifolds (possibly with  boundary); see~\cite[$\mathsection 1$]{Rourke1972IntroductionTP} for an introduction. We first generalize Theorem~\ref{t:homological-vKF}:

\begin{theorem}\label{t:forbidden}
For every integers $d \ge 3$ and $b$ there exists $N=N(d,b)$ such that $\Delta_{N}^{(\lceil d/2 \rceil)}$ does not homologically almost embed in any compact, $(\lceil d/2 \rceil-1)$-connected, $d$-dimensional PL manifold (possibly with boundary) $\M$ with $\beta_{\lceil d/2 \rceil}(\M;\Z_2) \le b$.
\end{theorem}

\noindent
This partially answers~\cite[Problem~3]{patakova2022bounding},~\cite[Conjecture~1.7]{steppingUp} and~\cite[Conjecture~2]{patak2025sharper} and extends the previous confirmation of Conjecture~\ref{conj:kalai-meshulam} from set systems in $\R^d$ to set systems on manifolds. This also extends several results on set systems with bounded topological complexity (Helly's theorem, Radon's theorem, $(p,q)$-theorem, \ldots) from $\R^d$ to sufficiently connected manifolds (see~\cite{patakova2022bounding,steppingUp}). We can relax the connectivity assumption (see Appendix~\ref{app:relax}) but not remove it.

\medskip

 One ingredient in the proof of Theorem~\ref{t:forbidden} is the following analogue of the Hanani-Tutte theorem for homological almost-embeddings, which is of independent interest.

\begin{theorem}\label{t:HT}
  Let $K$ be a simplicial complex of dimension $k>1$ and let $\M$ be a compact
  $2k$-dimensional PL manifold (possibly with boundary). If there exists a triangulation $T$ of $\M$ and a non-trivial chain map $f:C_*(K; \Z_2)\to C_*(T; \Z_2)$ in
  general position such that the images of any two non-adjacent
  $k$-faces $\sigma,\tau\in K$ intersect in an even number of points,
  then $K$ is a homological minor of $\mathcal{M}$.
\end{theorem}

\noindent
We formalize what we mean by general position in Section~\ref{s:back}.

\medskip

The homological shatter function $\Phi_\F^{(h)}$ defined in Equation~\eqref{eq:homshat} can be reformulated as a \emph{graded} version of the topological complexity $\hc^{(h)}_\F$ defined in Equation~\eqref{eq:topcomp}, where the value for parameter $t$ considers only intersections of subfamilies of size at most $t$: $\phi_\F^{(h)}(t)= \sup_{\substack{\F'\subset\F\\ |\F'| \le t} }\hc^{(h)}_{\F'}$. We systematize this viewpoint and define graded analogues of other parameters of set systems. For instance here are the \defn{graded Radon} and \defn{graded Helly} numbers:

\begin{equation}
  \rad_\F(t) \coloneqq \sup \limits_{\substack{\F'\subset \F \\ |\F'|\leq t}} \rad_{\F'}
 \quad \text{and} \quad
 \h_\F(t) \coloneqq \sup\limits_{\substack{\F'\subset \F \\ |\F'|\leq t}} \h_{\F'}.
\end{equation}

\noindent
It is straightforward to see that if the graded Helly numbers of a set system do not grow fast enough, then they are ultimately stationary and the set system has bounded Helly number (Lemma~\ref{l:hellygrowth}). We prove a similar condition for (graded) Radon numbers:

\begin{theorem}\label{t:radgrowth}
  Let $\F$ be a set system. If $\lim_{t \to \infty} \rad_\F(t) - \log_2 t = -\infty$, then $\rad_\F < \infty$.
\end{theorem}

\noindent
As a application, we extend Patáková's theorem from bounded to sufficiently slowly diverging homological shatter function (Corollary~\ref{c:Patakova+}). This, in turns yields fractional Helly theorems for sufficiently slowly diverging homological shatter functions.

\begin{corollary}\label{c:itgrows!}
  For every simplicial complex $K$ there exists a function $\Psi_K : \N \to \N$ with $\lim_{t \to \infty} \Psi_K(t) = +\infty$ such that the following holds. If $\F$ is a set system whose ground set has $K$ as forbidden homological minor and such that $\phi_{\F}^{(\dim K)}(t) \le \Psi_K(t)$ for $t$ large enough, then $\F$ has fractional Helly number at most $(\mu(K)+1)$.
\end{corollary}

\medskip

We then investigate more graded numbers and establish more relations between graded and ungraded numbers. As an application, we extend the Holmsen-Lee bound on $\Psi_{\rad \to \fh}(\cdot)$. Let $\Xi: \N \to \N$ be the function defined by $\Xi(r) := r^{r^{\lceil \log_2r \rceil}+r\lceil \log_2r \rceil}$. 

\begin{theorem}\label{t:radgrowth2}
  Let $\Psi: \N\rightarrow \N$ and $t_0 \in \N$ such that $\Psi(t) < t+1$ for every $t \ge t_0$. If there exists an integer $t_1 \ge {t_0}^2$ such that $\Xi\pth{\Psi(t_1)} < \frac{t_1}{t_0}$, then every set system whose graded Radon number function is bounded from above by $\Psi$ has bounded fractional Helly number.
\end{theorem}

\noindent
With Theorem~\ref{t:radgrowth2} we can strengthen Corollary~\ref{c:itgrows!}, as a close inspection of the proof reveals that the function $\beta$ associated to the fractional Helly number depends only on $\Psi$ and $t_0$. (It also allows a slightly faster growth than Corollary~\ref{c:Patakova+}.) We postpone this to the full version of the paper.

\section{On homological minors}\label{s:toolbox}\label{s:back}

In this section we prove Theorems~\ref{t:forbidden} and~\ref{t:HT}. For completeness, we start with a consequence of the simplicial approximation theorem that allows us to work purely in simplicial homology:

\begin{lemma}\label{l:going-generic-simplicial}
  A simplicial complex $K$ is a homological minor of a compact PL manifold (possibly with boundary) $\M$ if and only if $K$ is a homological minor of some triangulation of~$\M$.
\end{lemma}

\noindent
Due to space limitation we defer the proof of Lemma~\ref{l:going-generic-simplicial} to Appendix~\ref{app:simpapp}.

\medskip

When counting intersection points between chains, we focus on intersections that are stable under small perturbation. We therefore consider generic intersections, meaning intuitively that the chains intersect transversally. We formalize this in terms of linking numbers.\footnote{Intuitively, this generalizes the idea that in the plane, two curves cross at a point $x$ if the branches of the curves alternate around $x$.}

\medskip

Let $X$ be a triangulation of $\s^{2k-1}$. Let $S_1$ and $S_2$ be two subcomplexes of $X$ with $|S_1|$ and $|S_2|$ homeomorphic to $\s^{k-1}$. Suppose that the simplicial complex $X \setminus S_2$ induced by $X$ on the vertices not in $S_2$ has a geometric realization homotopy equivalent to $\s^{2k-1}\setminus \s^{k-1}$ (this can always be ensured up to taking a subdivision of $X$). The \defn{linking number} of $S_1$ and $S_2$ in~$X$ is $1$ if the homology class $[S_1]$ generates $H_{k-1}(X\setminus S_2;\Z_2)\cong\Z_2$, and $0$ if $[S_1]$ is trivial in $H_{k-1}(X\setminus S_2;\Z_2)$. (Exchanging $S_1$ and $S_2$ in this definition yields the same result.)

\medskip

Let $T$ be a triangulation of a compact PL manifold (possibly with boundary) $\M$. Two $k$-chains $z_1,z_2\in C_k(T;\Z_2)$ \defn{intersect generically} in vertex $v$ if the closed star $B$ of $v$ in $T$ satisfies: $B \cap z_1 \cap z_2 = \{v\}$, $D_1:= z_1 \cap B$ and $D_2:= z_2 \cap B$ are $k$-dimensional balls, and the spheres $\partial D_1$ and $\partial D_2$ have linking number $1$ in $\partial B$. Two $k$-chains $z_1,z_2\in C_k(T;\Z_2)$ are \defn{in general position} if $z_1 \cap z_2$ consists of finitely many vertices, and $z_1$ and $z_2$ intersect generically in each of these vertices. Two $k$-chains $z_1,z_2\in C_k(T;\Z_2)$ \defn{intersect evenly} if they are in general position and $\im(z_1)\cap \im(z_2)$ has even size.

\medskip

For $K$ a $k$-dimensional complex, a simplicial chain map $f:C_\bullet(K;\Z_2) \to C_\bullet(T;\Z_2)$ is in \defn{general position} if for every non-adjacent $k$-faces $\sigma, \tau \in K$, the chains $f(\sigma)$ and $f(\tau)$ are in general position and for each vertex $v \in f(\sigma) \cap f(\tau)$, $\sigma$ and $\tau$ are the only faces of $K$ whose images under $f$ intersect the closed star of $v$ in $T$. We say that a simplicial map $f:K \to T$ is  in \defn{general position} if the associated chain map $f_\#$ is.

\medskip

For any compact PL manifold (possibly with boundary) $\M$ of dimension $2k$, there exists a map $\cap_{\M}:H_k(\M;\Z_2)\times H_k(\M;\Z_2)\to\Z_2$, called the \defn{intersection form} of $\M$, such that $\cap_{\M}([z_1],[z_2])\in\Z_2$ counts the intersection points of $z_1$ and $z_2$ modulo $2$. We use no property of intersection forms besides their existence, and refer the interested reader to Prasolov~\cite[Chapter~2, $\S$2.7]{prasolov2007elements} for a precise definition and to Paták and Tancer~\cite{patak-tancer} for an brief account.

\medskip

The next result is due to Paták and Tancer~\cite[Proposition~21]{patak-tancer} and formulated following the presentation of Skopenkov~\cite{skopenkov}, which is better suited for our purpose. (More precisely, we reformulated \cite[Theorem~1.1.5]{skopenkov} using \cite[Lemma~2.1.1]{skopenkov}.) Note that the proofs by Paták-Tancer~\cite{patak-tancer} and by Skopenkov~\cite{skopenkov} use only PL maps.

\begin{theorem}
  \label{th:skopenkov}
  Let $L$ be a simplicial complex of dimension $k>1$ and let $\M$ be a compact, $(k-1)$-connected, $2k$-dimensional PL manifold (possibly with boundary). The following statements are equivalent:
  
  \begin{description}
  \item[(i)] There exists a triangulation $T$ of $\M$ and a simplicial map  $f:L \to T$ in general position such that the images of any two non-adjacent faces intersect evenly.\medskip
    
  \item[(ii)] There exists a triangulation $R$ of $\R^{2k}$ and a simplicial map $g: L \to R$ in general position and a map $\alpha$ that sends each $k$-face of $L$ to an element of $H_k(\M;\Z_2)$ such that any two non-adjacent $k$-faces $\sigma,\tau \in L$ have images that intersect in an even number of points if and only if $\cap_{\M}(\alpha(\sigma),\alpha(\tau)) = 0$.
  \end{description}
\end{theorem}

\subsection{A homological Hanani-Tutte theorem} 

Let us now prove Theorem~\ref{t:HT}. Let $K$ be a simplicial complex of dimension $k > 1$ and let $T$ be a triangulation of a compact PL manifold (possibly with boundary) $\M$ of dimension $2k$. Let $f:C_*(K;\Z_2)\to C_*(T; \Z_2)$ be a non-trivial chain map in general position such that the images of any two non-adjacent $k$-faces $\sigma,\tau\in K$ intersect evenly.

\medskip

Our goal is to prove that $K$ is a homological minor of $\M$. This requires repeatedly subdividing the triangulation $T$. In what follows, every time we subdivide a triangulation $T_i$ into $T_{i+1}$, all chain maps to $T_i$ and subcomplexes of $T_i$ are also subdivided to $T_{i+1}$. Also, throughout the proof we identify every pure $\ell$-dimensional simplicial complex with the unique $\ell$-chain with coefficients in $\Z_2$ it supports; we abuse the terminology and say that we add (pure) simplicial complexes over $\Z_2$ to mean that we add the corresponding chains.

\medskip

If $f$ is a homological almost-embedding we are done. Otherwise, there exist some non-adjacent $k$-faces $\sigma,\tau \in K$ such that $f(\sigma)\cap f(\tau)$ is non-empty. By assumption, this intersection is a set of vertices of even size, so let $x,y\in f(\sigma)\cap f(\tau)$ be two distinct such vertices. Recall that $f$ is in general position, $f(\sigma)$ and $f(\tau)$ intersect generically in $x$ and in~$y$. We set out to construct a refinement $T'$ of $T$ and a new map $f':C_*(K; \Z_2)\to C_*(T'; \Z_2)$ that is also in general position, differs from $f$ only on $\sigma$ and $\tau$, and satisfies $f'(\sigma)\cap f'(\tau) = \pth{f(\sigma)\cap f(\tau)} \setminus \{x,y\}$ as well as $f'(\alpha)\cap f'(\beta) = f(\alpha)\cap f(\beta)$ for any $\alpha \in \{\sigma,\tau\}$ and any $k$-face $\beta \in K \setminus \{\sigma,\tau\}$. Iterating this procedure produces the announced homological almost-embedding of $K$ into a triangulation of~$\mathcal{M}$.

\begin{figure}[ht]
  \begin{center}\includegraphics[width=0.7\textwidth,keepaspectratio,page=1]{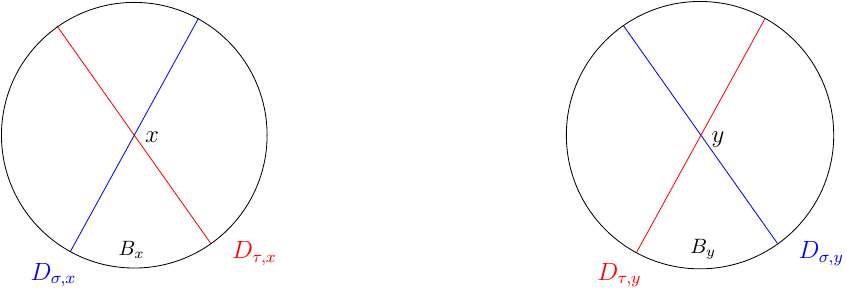}
  \end{center}
  \caption{The setup for the construction of $T'$ and $f'$.}
\end{figure}

Let $B_x$ and $B_y$ denote the closed stars of $x$ and $y$ in $T$. For $z \in \{x,y\}$ and $\alpha \in \{\sigma, \tau\}$ let $D_{\alpha,z}:= \supp(f(\alpha))\cap B_z$. Since $f(\sigma)$ and $f(\tau)$ intersect generically in $x$ and $y$, the  complexes  $D_{\sigma,x}, D_{\sigma,y}, D_{\tau,x}$ and $D_{\tau,y}$ are $k$-dimensional balls, the $(k-1)$-spheres $\partial D_{\sigma,x}$ and $\partial D_{\tau,x}$ have linking number $1$ in $\partial B_x$, and similarly $\partial D_{\sigma,y}$ and $\partial D_{\tau,y}$ have linking number~$1$ in $\partial B_y$.

\medskip  

We subdivide $T$ into $T_1$ so that the closed star $B'_x$ of $x$ and $B_y'$ of $y$ in $T_1$ are disjoint from $\partial B_x$ and $\partial B_y$, respectively. In particular, for $z \in \{x,y\}$ and $\alpha \in \{\sigma, \tau\}$, letting $D_{\alpha,z}':= \supp(f(\alpha))\cap B_z'$, the pair $(\partial D_{\alpha,z}', \partial B'_z)$ is homeomorphic to $(\s^{k-1},\s^{2k-1})$. The genericity of $x$ and $y$ in $f(\sigma) \cap f(\tau)$ is preserved through the subdivision $T \to T_1$ so  the four complexes  $D_{\bullet,\bullet}'$ are $k$-dimensional balls and their bounding spheres have the same linking numbers in $\partial B_x'$ and $\partial B_y'$ as their counterparts on $\partial B_x$ and $\partial B_y$.

\begin{figure}[ht]
  \begin{center} \includegraphics[width=0.7\textwidth,keepaspectratio,page=2]{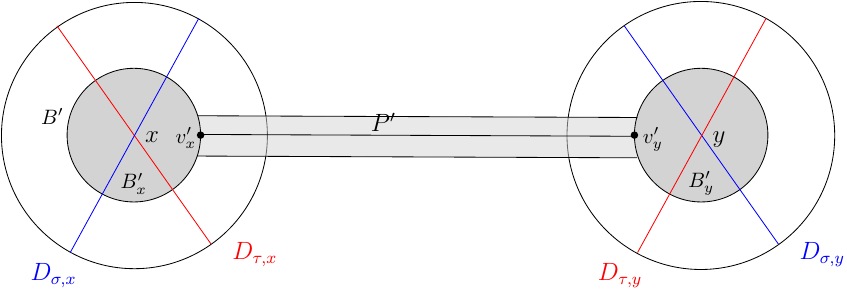}
  \end{center}
  \caption{The construction of $B'$.}
\end{figure}

Up to subdividing $T_1$ into $T_2$, there exist vertices $v_x'$ and $v_y'$ in $\partial B_x'$ and $\partial B_y'$, respectively, and a path $P'$ in the 1-skeleton of $T_2$ from $v_x'$ to $v_y'$ such that, letting $N(P')$ denote the closed star of $P'$ in $T_2$, the closed star of $N(P')$ is disjoint from the image of $f$. The union $B'  := B_x'\cup N(P') \cup B_y'$ is a ball. Observe that in $\partial B'$, the $(k-1)$-spheres $\partial D_{\sigma,x}'$ and $\partial D_{\tau,x}'$ retain the linking number $1$ that they have on $\partial B_x'$. Similarly, in $\partial B'$, the $(k-1)$-spheres $\partial D_{\sigma,y}'$ and $\partial D_{\tau,y}'$ retain the linking number $1$ that they have on $\partial B_y'$. 

\begin{figure}[ht]
  \begin{center}  \includegraphics[width=0.7\textwidth,keepaspectratio,page=3]{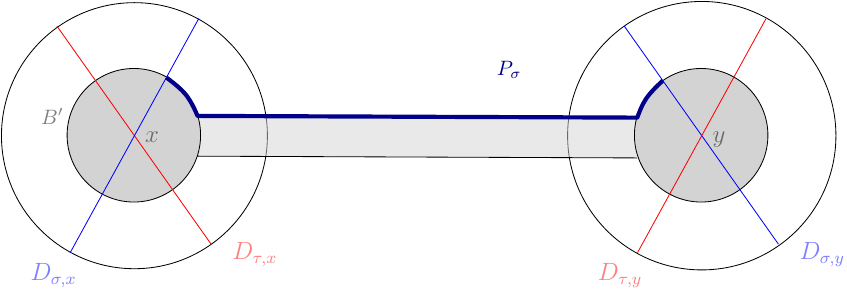}
  \end{center}
  \caption{The path $P_\sigma$.}
\end{figure}

Thus, up to further subdividing $T_2$ into $T_3$, there exist a path $P_\sigma$ in the 1-skeleton of $\partial B'$ that connects $\partial D_{\sigma,x}'$ to $\partial D_{\sigma,y}'$ and with relative interiors disjoint from the image of $f$. (Here we use that $\partial D_{\alpha,z}'$ is of codimension $k \ge 2$ in $\partial B'$.) Up to subdividing the triangulation further, we can pipe $D_{\sigma,x}$ and $D_{\sigma,y}$ together by a tube $F_{\sigma}$ found in a neighborhood of the path $P_{\sigma}$~\cite[$\S 5.10$]{Rourke1972IntroductionTP}. The tube $\pth{F_\sigma, F_\sigma \cap D_{\sigma,x}, F_\sigma \cap D_{\sigma,y}}$ is homeomorphic to $\pth{\mathbb{D}^{k}\times [0,1], \mathbb{D}^{k}\times\{0\}, \mathbb{D}^{k}\times\{1\}}$. Moreover, by the (PL) general position theorem for embeddings~\cite[$\S$5.3]{Rourke1972IntroductionTP}, we can take $F_{\sigma}$ such that $\partial B'$ and $\partial F_{\sigma}$ intersect in a generic way. By taking $F_{\sigma}$ in a sufficiently small neighborhood $U$ of $P_{\sigma}$ so that $(U,\partial B'\cap U)$ is homeomorphic to $(\R^{2k},\R^{2k-1}\times\set{0})$, and the triple $\pth{F_\sigma \cap \partial B', F_\sigma \cap D_{\sigma,x}  \cap \partial B', F_\sigma \cap D_{\sigma,y}  \cap \partial B'}$ is homeomorphic to $\pth{\mathbb{D}^{k-1}\times [0,1], \mathbb{D}^{k-1}\times\{0\}, \mathbb{D}^{k-1}\times\{1\}}$.

\begin{figure}[ht]
  \begin{center}  \includegraphics[width=0.7\textwidth,keepaspectratio,page=4]{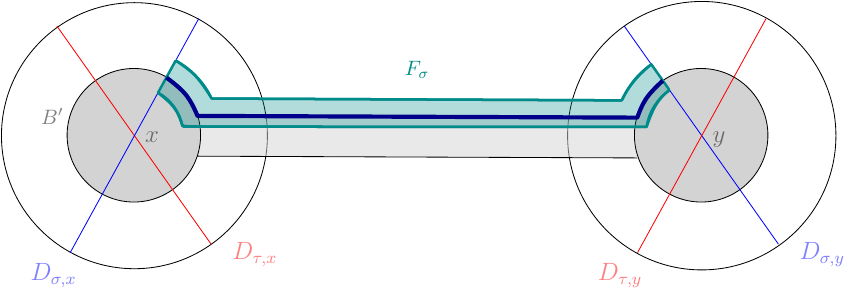}
  \end{center}
  \caption{The piping $F_\sigma$ between $D_{\sigma,x}$ and $D_{\sigma,y}$ in a neighborhood of $P_{\sigma}$.}
\end{figure}

Now, let $C_\sigma$ be the sum $D_{\sigma,x} + \partial F_\sigma + D_{\sigma,y}$ (over $\Z_2$). Note that $C_\sigma$ is contained in the union of $B_x \cup B_y$ and the closed star of $N(P')$, and therefore intersects the image of $f$ only inside $B_x \cup B_y$. We note that $F_\sigma \cap \partial B'$ also pipes the $(k-1)$-spheres $D_{\sigma,x}  \cap \partial B'$ and $D_{\sigma,y}  \cap \partial B'$. Indeed, $\pth{F_\sigma \cap \partial B', F_\sigma \cap D_{\sigma,x}  \cap \partial B', F_\sigma \cap D_{\sigma,y}  \cap \partial B'} = (F_\sigma \cap \partial B', F_\sigma \cap \partial(D_{\sigma,x}  \cap B')$, $F_\sigma \cap \partial(D_{\sigma,y}  \cap  B'))$ is homeomorphic to $\pth{\mathbb{D}^{k-1}\times [0,1], \mathbb{D}^{k-1}\times\{0\}, \mathbb{D}^{k-1}\times\{1\}}$. It follows that $C_\sigma \cap \partial B'$ is homeomorphic to the connected sum of two  $(k-1)$-spheres, and is therefore homeomorphic to a $(k-1)$-sphere.

\begin{figure}[ht]
  \begin{center}   \includegraphics[width=0.7\textwidth,keepaspectratio,page=5]{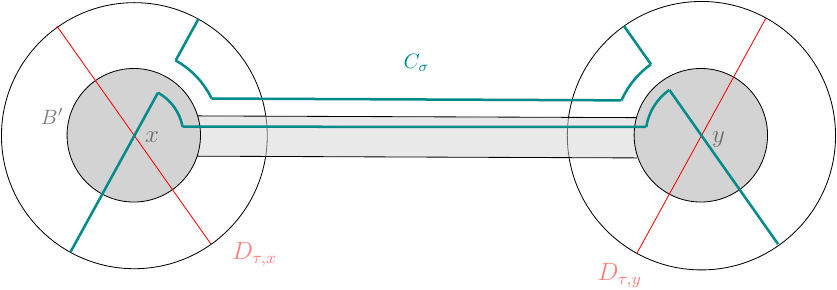}
  \end{center}
  \caption{The chain $C_\sigma = D_{\sigma,x} + \partial F_\sigma + D_{\sigma,y}$ over $\Z_2$.}
\end{figure}

We claim that, in $\partial B'$, the linking number between $C_{\sigma}\cap\partial B'$ and $D_{\tau,x}\cap\partial B'$ equals the linking number between $D_{\sigma,x}\cap\partial B'$ and $D_{\tau,x}\cap\partial B'$, that is $1$. This follows from the fact that the chain $C_\sigma$ differs from $D_{\sigma,x}$ by $\partial F_\sigma + D_{\sigma,y}$, and that each of $\partial F_{\sigma}\cap\partial B'$ and $D_{\sigma,y}\cap\partial B'$ is a boundary in $\partial B'\setminus D_{\tau,x}$. The same claim holds if we exchange $x$ for $y$. 

\begin{figure}[ht]
  \begin{center} \includegraphics[width=0.7\textwidth,keepaspectratio,page=6]{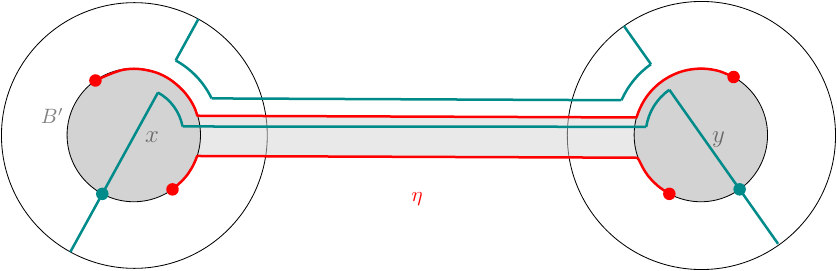}
  \end{center}
  \caption{The chain $\eta\in C_k(\partial B'\setminus C_{\sigma};\Z_2)$ such that $\partial \eta = D_{\tau,x}\cap\partial B'+ D_{\tau,y}\cap\partial B'$ over $\Z_2$.}
\end{figure}

Altogether, we get that $D_{\tau,x}\cap\partial B'$ and $D_{\tau,y}\cap\partial B'$ are in the same homology class in $\partial B'\setminus C_{\sigma}$. Hence, their sum is a boundary, and there exists a chain $\eta\in C_k(\partial B'\setminus C_{\sigma};\Z_2)$ such that the support of $\partial \eta$ is the sum over $\Z_2$ of $D_{\tau,x}\cap\partial B'$ and $D_{\tau,y}\cap\partial B'$. We finally set 
\[ f'(\sigma) = \hspace{-0.9cm}\underbrace{f(\sigma)+D_{\sigma,x} + D_{\sigma,y}}_{f(\sigma) \text{ with its restriction to $B_x \cup B_y$ removed}} \hspace{-1cm}+ C_{\sigma} \qquad \text{and} \qquad f'(\tau)= \hspace{-0.9cm} \underbrace{f(\tau)+D_{\tau,x}'+D_{\tau,y}'}_{f(\tau) \text{ with its restriction to $B'$ removed}}\hspace{-0.9cm}+\eta\]

\begin{figure}[ht]
  \begin{center}  \includegraphics[width=0.7\textwidth,keepaspectratio,page=7]{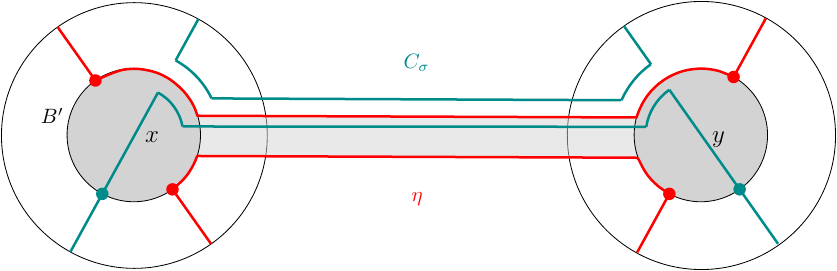}
  \end{center}
  \caption{The chains $C_\sigma$ and $\eta$ reroute $f(\sigma)$ and $f(\tau)$ so as to remove intersections in $B_x$ and~$B_y$.}
\end{figure}

\noindent
We set $f'(\omega)=f(\omega)$ for every other face $\omega\in K$. Extending $f'$ linearly yields a chain map, since both $f'(\sigma)-f(\sigma)$ and $f'(\tau)-f(\tau)$ are cycles \footnote{They are even boundaries, ensuring that $f'$ is chain homotopic to $f$}. The chain map $f'$ is as announced: $f'(\sigma)\cap f'(\tau) = f(\sigma)\cap f(\tau) \setminus \{x,y\}$ and every other intersection remains unchanged. This concludes the proof of Theorem~\ref{t:HT}.

\subsection{Forbidden homological minors for manifolds}

We now prove Theorem~\ref{t:forbidden}. First, note that the odd-dimensional case reduces to the even-dimensional one. Indeed, for every $k \ge 2$, if $\M$ is a compact, $(k-1)$-connected, $(2k-1)$-dimensional PL manifold (possibly with boundary), then $\M \times [0,1]$ is a compact, $(k-1)$-connected, $2k$-dimensional PL manifold (possibly with boundary). Moreover, $\beta_{k}(\M \times [0,1];\Z_2) = \beta_{k}(\M;\Z_2)$ and any homological minor of $\M$ is a homological minor of $\M \times [0,1]$. If the statement holds for $d$ even and $b \in \N$ with some $N(b,d)$, then it holds with $d-1$ and $b$ by putting $N(d-1,b) := N(d,b)$. So let us now consider the even-dimensional case.

\medskip

 Let $k \ge 2$ and $b \in \N$, and let $\M$ be a compact, $(k-1)$-connected, $2k$-dimensional PL manifold (possibly with boundary). Let us fix $N$ and suppose that $K = \Delta_N^{(k)}$ is a homological minor of $\M$. By Lemma~\ref{l:going-generic-simplicial}, there exist a triangulation $T$ of $\M$ and a (simplicial) homological almost-embedding $C_*(K)\to C_*(T)$. Actually, letting $S := T^{(k)}$, we have that there exists a homological almost-embedding $a:C_*(K;\Z_2) \to  C_*(S;\Z_2)$.

  \medskip

  We apply Theorem~\ref{th:skopenkov} with $L=S$. Condition~(i) holds with $T=S$ and $f$ the identity, so Condition~(ii) also holds. Hence, there exists a triangulation $R$ of $\R^{2k}$, a simplicial map $g:S \to R$ in general position, and a map $\alpha$ that sends each $k$-face of $S$ to an element of $H_k(\M;\Z_2)$ such that any two non-adjacent $k$-faces $\sigma,\tau \in S$ intersect evenly if and only if $\cap_{\M}(\alpha(\sigma),\alpha(\tau)) = 0$. We let $\tilde{\alpha}:C_k(S) \to H_k(\M;\Z_2)$ denote the linear extension of $\alpha$.

\medskip

Consider the chain map $b:C_*(K;\Z_2) \to C_*(R;\Z_2)$ defined by $b = g_\# \circ a$. Note that $b$ is nontrivial since $a$ is nontrivial and $g$ is a simplicial map. Moreover, the fact that $b$ is a chain map in general position follows from three observations:

\begin{itemize}
\item since $g$ is a simplicial map in general position, $g_\#$ is a chain map in general position,
\item since $a$ is a homological almost-embedding, it is also a chain map in general position, and 
\item the composition of a homological almost-embedding and a chain map in general position is a chain map in general position.
\end{itemize}

\noindent
Notice that for $N \ge 2k+3$, not every independent $k$-faces of $K$ can have images under $b$ that intersect evenly. Indeed Theorem~\ref{t:HT} would then imply that $K$ is a homological minor of $\R^{2k}$, which would contradict the homological van Kampen-Flores theorem (Theorem~\ref{t:homological-vKF}). We use Ramsey's theorem to show that this contradiction can be reached for some subcomplex of $K$.

\medskip

So consider two non-adjacent $k$-faces $\sigma,\tau \in K$ and put $a(\sigma) = \sigma_1+\sigma_2+\ldots +\sigma_s$ and $a(\tau) =\tau_1+\tau_2+\ldots+\tau_t$. Since $a$ is a homological almost-embedding, $\supp(a(\sigma))$ and $\supp(a(\tau))$ are disjoint, and $\sigma_i$ and $\tau_j$ are thus non-adjacent for every $(i,j) \in [s] \times [t]$. Hence, given $(i,j) \in [s] \times [t]$, $g(\sigma_i)$ and $g(\tau_j)$ intersect evenly if and only if $\cap_{\M}(\alpha(\sigma_i),\alpha(\tau_j)) = 0$. We can thus count the intersections of $b(\sigma)$ and $b(\tau)$ (the following equalities are modulo $2$ and $\card{}$ denotes the cardinal):

\[\begin{aligned}
\card{b(\sigma) \cap b(\tau)} = &\sum_{i \in [s], j \in [t]} \card{g(\sigma_i) \cap g(\tau_j)}\\
= &\sum_{i \in [s], j \in [t]} \cap_{\M}(\alpha(\sigma_i),\alpha(\tau_j))\\
= &\quad\sum_{i \in [s]} \cap_{\M}(\alpha(\sigma_i),\tilde{\alpha}(a(\tau)))= \cap_{\M}(\tilde{\alpha}(a(\sigma)),\tilde{\alpha}(a(\tau)))
\end{aligned}  \]

\medskip

Let $r = \beta_k(\M;\Z_2)$. The map $\beta: C_k(K) \to H_k(\M;\Z_2)$ defined by $\beta = \tilde{\alpha} \circ a$ induces a coloring of the $k$-simplices of $K$ by the (at most $2^r$) elements of $H_k(M;\Z_2)$. Let $N'$ denote the number of vertices of $\sd \Delta_{2k+2}^{(k)}$. By the hypergraph Ramsey theorem, for $N$ large enough (as a function of $r$ and $k$) there exists a subset $W$ of $N'$ vertices in $K$ such that $\beta$ is constant over all $k$-simplices of $K[W]$; let us denote by $\square$ this constant value. In particular, for every $k$-faces $\sigma,\tau$ of $K[W]$ such that $\sigma$ and $\tau$ are not adjacent, we have, modulo $2$, $\card{b(\sigma) \cap b(\tau)} = \cap_{\M}(\beta(\sigma_1),\beta(\tau_1)) = \cap_{\M}(\square,\square)$. Let us fix a bijection from the vertices of $\sd\Delta_{2k+2}^{(k)}$ to $W$ and extend it to a chain map $j:C_*\pth{\sd\Delta_{2k+2}^{(k)}}\to C_*(K[W])$. Also, let $h:C_*\pth{\Delta_{2k+2}^{(k)}}\to C_*\pth{\sd \Delta_{2k+2}^{(k)}}$ denote the subdivision chain map, where each $i$-face of $\Delta_{2k+2}^{(k)}$ is mapped to the sum of the $i$-faces of $\sd \Delta_{2k+2}^{(k)}$ that it contains. In particular, for every $k$-face $\sigma$ of $\Delta_{2k+2}^{(k)}$, the chain $h(\sigma)$ is supported on $k!$ $k$-faces of $\sd \Delta_{2k+2}^{(k)}$.

\medskip
 
Let us examine the properties of $b\circ j \circ h$. First, it is a nontrivial chain map (because $b$, $j$ and $h$ are). Moreover, $j \circ h$ is a homological almost-embedding (since $j$ and $h$ are), and its composition with the chain map in general position $b$ yields a chain map in general position. Furthermore, any two non-adjacent $k$-faces $\sigma,\tau \in \Delta_{2k+2}^{(k)}$ have images under $b\circ j \circ h$ that intersect evenly. To see this, let us put $j \circ h(\sigma) = \sigma_1+\sigma_2+\ldots +\sigma_s$ and $j \circ h(\tau) = \tau_1+\tau_2+\ldots+\tau_s$. Since $j \circ h$ is a homological almost-embedding, $\supp(j \circ h(\sigma))$ and $\supp(j \circ h(\tau))$ are disjoint. It follows that for every $i,j \in [s]$ the $k$-simplices $\sigma_i$ and $\tau_j$ are non-adjacent. We therefore have, modulo 2, 

\[\begin{aligned}
\card{b\circ j \circ h(\sigma) \cap b\circ j \circ h(\tau)} &= \sum_{i,j \in [s]} \card{b(\sigma_i) \cap b(\tau_j)}\\
&= \sum_{i,j \in [s]} \cap_{\M}(\beta(\sigma_i),\beta(\tau_j)) = \sum_{i,j \in [s]} \cap_{\M}(\square,\square).
\end{aligned}  \]

\noindent
To sum up, the cardinal of $b\circ j \circ h(\sigma) \cap b\circ j \circ h(\tau)$ has the same parity as $\cap_{\M}(\square,\square)s^2$. Since $s=k!$ is even, $\sigma$ and $\tau$ have images under $b\circ j \circ h$ that intersect evenly.

\medskip

To conclude, $b\circ j \circ h$ is a nontrivial chain map in general position from $\Delta_{2k+2}^{(k)}$ to $R$ such that independent faces have images that intersect evenly. By Theorem~\ref{t:HT} this means that $\Delta_{2k+2}^{(k)}$ is a homological minor of $\R^{2k}$, a contradiction with Theorem~\ref{t:homological-vKF}. Thus, for $N$ large enough, the initial hypothesis that $\Delta_N^{(k)}$ is a homological minor of $\M$ cannot be true.

\section{Graded parameters of set systems}\label{s:graded}

In this section we present our contributions on set systems of parameters. 

\subsection{Graded Radon and Helly numbers}\label{s:graded1}

Each relation between parameters of a set system yields a relation between their graded analogues. From Levi's inequality we get the following inequality between the graded Helly and Radon numbers (defined in Section~\ref{s:cbk}):

\begin{equation}\label{eq:gradedhr}
  \forall t \in \N, \qquad \h_\F(t) =  \sup\limits_{\substack{\F'\subset \F \\ |\F'|\leq t}} \h_{\F'} \le \sup\limits_{\substack{\F'\subset \F \\ |\F'|\leq t}} \pth{\rad_{\F'}-1} = \rad_\F(t) - 1. 
\end{equation}

It follows from the definitions of graded parameters that each one is a non-decreasing function that converges to the ungraded parameter (possibly $\infty$). We notice that if a graded Helly number is asymptotically sublinear then it is bounded: 

\begin{lemma}\label{l:hellygrowth}
  Let $\F$ be a set system and $t_0 \in \N$. If $\h_\F(t) < t$ for all $t>t_0$, then $\h_\F \leq t_0$.
\end{lemma}
\begin{proof}
  By definition, we have $\h_\F(t) \le t$ for every $t \in \N$. Moreover, $\h_\F(t) \neq \h_\F(t-1)$ if and only if $\h_\F(t) = t$. The assumption and a straightforward induction therefore implies that for every $t>t_0$ we have $\h_\F(t) = \h_\F(t_0) \le t_0$.
\end{proof}

\noindent
When graded Radon numbers grow, it is at most linearly (see Appendix~\ref{app:rab}) and not too slowly:

\begin{lemma}\label{l:radongrowth}
  Let $\F$ be a set system and $t\ge 2$ an integer. If $\rad_\F(t) > \rad_\F(t-1)$, then $\rad_\F(t-1) \ge 1+\log_2\pth{1+\frac{t}{\h_\F(t)}}$.
\end{lemma}
\begin{proof}
  Let $X$ denote the ground set of $\F$ and let $n = \rad_\F(t-1)$. Suppose that $\rad_\F(t)>n$, so that there exist a subset $\G = \{G_1,G_2, \ldots, G_t\} \subset \F$ and a subset $S \subseteq X$ of size $n$ such that

\begin{description}
\item[(i)] there is no partition of $S$ into two parts whose $\G$-convex hulls intersect,\medskip
\item[(ii)] for every $i \in [t]$, there exists a partition $\p_i$ of $S$ into two parts whose $(\G\setminus \{G_i\})$-convex hulls intersect.
\end{description}

\noindent
Recall that given $\F' \subseteq \F$, the $\F'$-convex hull $\conv_{\F'}(P)$ of a subset $P\subset X$ is the intersection of all the members of $\F'$ that contain $P$. In particular, for any $P \subseteq X$ such that $P \not \subseteq G_{i}$ we have $\conv_\G(P) = \conv_{\G\setminus \{G_{i}\}}(P)$. Conditions~(i) and~(ii) therefore imply that every $G_i \in \G$ contains one or the other part of $\p_i$.

\medskip

There are at most $2^{n-1}-1$ partitions of~$S$ in two nonempty parts. Let us assume that $t > (2^{n-1}-1)h$ for some integer $h$, so that by the pigeonhole principle there exist $h+1$ indices $i_1,i_2, \ldots, i_{h+1}$ such that the partitions $\p_{i_1}, \p_{i_2}$, \ldots, $\p_{i_{h+1}}$ coincide. Let $\{P_1,P_2\}$ be that partition of $S$. Let us put $\G' = \{ A \in \G \colon P_1 \subseteq A \text{ or } P_2 \subseteq A\}$. We make two observations:

\begin{itemize}
\item $\cap_{A \in \G'}A$ coincides with $\conv_{\G}(P_1) \cap \conv_{\G}(P_2)$ and is therefore empty.\medskip
\item every choice of $h$ elements in $\G'$ has nonempty intersection. Indeed, $\G'$ contains $G_{i_1}$, $G_{i_2}$, \ldots, $G_{i_{h+1}}$, so that any choice of $h$ elements from $\G'$ is bound to miss $G_{i_j}$ for at least one $j \in [h+1]$ and their intersection must then contain $\conv_{\G \setminus \{G_{i_j}\}}(P_1) \cap \conv_{\G \setminus \{G_{i_j}\}}(P_2) \neq \emptyset$. 
\end{itemize}

\noindent
For $h \ge \h_\F(t)$ these conditions are incompatible. We therefore have $t \le (2^{n-1}-1) \h_\F(t)$, and the statement follows.
\end{proof}

We can now prove that any set system with sufficiently slowly growing graded Radon numbers has finite Radon number.

\begin{proof}[Proof of Theorem~\ref{t:radgrowth}]
  Let $\F$ be a set system with infinite Radon number. If the Helly number $\h_\F$ is also infinite, then by Lemma~\ref{l:hellygrowth} there exists an increasing sequence $\{t_i\}_{i \in \N}$ such that $\h_\F(t_i) = t_i$, and therefore $\rad_\F(t_i) \ge t_i+1$ by Inequality~\eqref{eq:gradedhr}; this prevents $\rad_\F(t)-\log_2 t$ from going to $-\infty$ as $t \to \infty$. So suppose that the Helly number $\h_\F$ is finite. The assumption that $\rad_\F=\infty$ ensures that there exists an increasing sequence $\{t_i\}_{i \in \N}$ such that $\rad_\F(t_i) > \rad_\F(t_i-1)$. Lemma~\ref{l:radongrowth} implies that $\rad_\F(t_i) > \rad_\F(t_i-1) \ge \log_2 t_i - \log_2 \h_\F$. Again, this prevents $\rad_\F(t)-\log_2 t$ from going to $-\infty$ as $t \to \infty$. The statement follows by contraposition.
\end{proof}

\subsection{Consequences for topological set systems}

Let us finally consider topological set systems with slowly growing homological shatter function and ground set with a forbidden homological minor.

\begin{corollary}\label{c:Patakova+}
  For every simplicial complex $K$ there exists a function $S_K : \N \to \N$ with $\lim_{t \to \infty} S_K(t) = +\infty$ such that the following holds. Any set system $\F$ whose ground set has $K$ as forbidden homological minor and satisfies $\phi_{\F}^{(\dim K)}(t) \le S_K(t)$ for $t$ large enough has finite Radon number.  
\end{corollary}
\begin{proof}
  Recall that $\Psi^{(K)}_{\hc \to \rad}(x)$ denotes the supremum of the Radon number of a set system with $(\dim K)$-level topological complexity at most $x$ and whose ground set has $K$ as forbidden homological minor. Patáková~\cite{patakova2022bounding} proved that $\Psi^{(K)}_{\hc \to \rad}(x)$ is finite for every $K$ and $x$.

  \medskip

  For $t \in \N$, we define $S_K(t) =\max\{x\in\N\mid \Psi^{(K)}_{\hc \to \rad}(x) \le \frac12 \log_2 t\}$. This ensures that $\Psi^{(K)}_{\hc \to \rad}(S_K(t)) \le \frac12 \log_2 t$ for every $t \in \N$.    Observe that $\lim_{t \to \infty} S_K(t) = \infty$ since $\Psi^{(K)}_{\hc \to \rad}(x)$ is finite for every $x \in \N$.

  \medskip

  Now consider a set system $\F$ with function $\phi_{\F}^{(\dim K)} \le S_K$ and whose ground set has $K$ as forbidden homological minor. Let $t \in \N$ and consider a subset $\F' \subseteq \F$ of size $t$. The ground set of $\F'$ also has $K$ as forbidden homological minor. Moreover, $\F'$ has  $(\dim K)$-level topological complexity at most $\phi_{\F}^{(\dim K)}(t) \le S_K(t)$. It follows that $\rad_{\F'} \le \Psi^{(K)}_{\hc \to \rad}(S_K(t)) \le \frac12 \log_2 t$. This holds for every $\F' \subseteq \F$ of size $t$, so $\rad_\F(t) \le \frac12 \log_2 t$. This inequality holds for every $t \in \N$ so Theorem~\ref{t:radgrowth} implies that $\F$ has bounded Radon number.
\end{proof}

\noindent
We can finally prove a fractional Helly theorem for diverging homological shatter functions.

\begin{proof}[Proof of Corollary~\ref{c:itgrows!}]
  Recall that $\Psi_{\rad \to \fh}(y)$ denotes the supremum of the fractional Helly number of a set system with Radon number $y$. Holmsen and Lee~\cite[Theorem~1.1]{Holmsen2021} proved that $\Psi_{\rad \to \fh}(y)$ is finite for every $y$. Let $S_K$ denote the function from Corollary~\ref{c:Patakova+}. Now consider a set system $\F$ whose ground set has $K$ as forbidden homological minor and satisfies $\phi_{\F}^{(\dim K)} \le S_K$. Corollary~\ref{c:Patakova+} ensures that $\rad_\F$, the Radon number of $\F$, is finite. It follows that the fractional Helly number of $\F$ is at most $\Psi_{\rad \to \fh}(\rad_\F)$, and is therefore finite. From there, \cite[Theorem~1.2]{steppingUp} ensures that this fractional Helly number is at most $\mu(K)+1$.
\end{proof}

\subsection{Other graded parameters and relations}\label{s:graded2}

With the intersection hypergraph of $\F$ in mind, we say that a set $\G$ in a set system $\F$ is a \defn{clique} if the intersection of all members of $\G$ is non-empty. The colorful Helly theorem suggests the following parameter:

\begin{itemize}
\item The \defn{colorful Helly number} $\ch_\F$ of $\F$ is the smallest number of colors $m$ such that for every coloring of a subfamily $\F'\subset \F$ with $m$ colors, if every subfamily that contains exactly one element of each color is a clique, then at least one color class is a clique.
\end{itemize}

\noindent
We say that a set $\G$ in a set system $\F$ is a \defn{$c$-wise clique} if every $c$-element subset of $\G$ is a clique. Clearly a clique is a $c$-wise clique, and when $c \ge \h_\F$ the converse is true. To analyze  set systems with large, infinite or unknown Helly numbers, it is useful to consider variants of the colorful and fractional Helly numbers where cliques are replaced by $c$-wise cliques:

\begin{itemize}
\item The \defn{$c$th colorful Helly number} $\ch^{(c)}_\F$ of $\F$ is the smallest number of colors $m \geq c$ such that for every coloring of a subfamily $\F'\subset \F$ with $m$ colors, if every subfamily of $\F'$ that contains exactly one element of each color forms a $c$-wise clique, then at least one color class is a $c$-wise clique.\medskip
 
\item The \defn{$c$th fractional Helly number} $\fh^{(c)}_\F$ of $\F$ is the smallest integer $s$ such that there exists a function $\beta_\F:(0,1)\rightarrow (0,1)$ with the following property: For every finite subfamily $\F'\subset \F$, whenever a fraction $\alpha$ of the $s$-tuples of $\F'$ forms a $c$-wise clique, a subset $\G$ of $\F'$ of size $\beta_\F(\alpha)|\F'|$ forms a $c$-wise clique. \medskip
\end{itemize}

\noindent
Obviously for every set system $\F$, if $c \ge \h_\F$, then $\fh^{(c)}_\F = \fh_\F$ and $\ch^{(c)}_\F = \ch_\F$. Holmsen~\cite[Theorem~1.2]{Holmsen2020} proved that $\fh^{(c)}_\F\leq \ch^{(c)}_\F$ for every set system $\F$ and every $c\in\N$. A close inspection of that proof provides another bridge between the graded and ungraded parameters:

\begin{lemma}\label{l:holmsen-graded}
Let $\ell > c$ be integers. Every set system $\F$ such that $\ch^{(c)}_\F(c\ell) \leq \ell$ satisfies $\fh^{(c)}_\F \leq \ch^{(c)}_\F(c\ell)$.
\end{lemma}
\begin{proof}
  Let $c< \ell$, let $\F$ be a set system and let $\ell':=\ch^{(c)}_\F(c\ell)$. By definition of $\ch^{(c)}_\F(\cdot)$, the $c$-uniform hypergraph recording which $c$-element subsets of $\F$ form cliques cannot contain a certain pattern on $c\ell'$ vertices, namely the \emph{complete $\ell'$-tuples of missing edges}~\cite[$\mathsection 3$]{Holmsen2021}. This is the only property needed to ensure that $\fh^{(c)}_\F \leq \ell'$~\cite[Theorem~1.2]{Holmsen2020}.
\end{proof}

\medskip

  Let $\Psi_{\rad \to \ch}^{(c)}(x)$ denote the supremum of the $c$th colorful Helly number of a set system with Radon number $x$. Holmsen and Lee~\cite[Theorem~2.2]{Holmsen2021} proved that $\Psi_{\rad \to \ch}^{(c)}(x) \le \max\pth{\Xi(x),c}$ for every $c \ge x-1$, where  $\Xi: \N \to \N$ is the function defined by $\Xi(r) := r^{r^{\lceil \log_2r \rceil}+r\lceil \log_2r \rceil}$. Applying this inequality to subsets of size $t$ we get:
\begin{align}\label{eq:colorful-helly-g}
  \ch^{(c)}_\F(t) \leq \max\pth{\Xi\pth{\rad_\F(t)},c} && \text{for every } \F, c, t \text{ such that } c \ge \h_\F(t).
\end{align}

\medskip

We can now prove that sufficiently slowly growing graded Radon numbers imply a bounded Fractional Helly number.

\begin{proof}[Proof of Theorem~\ref{t:radgrowth2}]
  Let $\Psi: \N\rightarrow \N$ and $t_0 \in \N$ such that $\Psi(t) < t+1$ for every $t \ge t_0$. Also suppose that there exists an integer $t_1 \ge {t_0}^2$ such that $\Xi\pth{\Psi(t_1)} < \frac{t_1}{t_0}$. We now consider a set system $\F$ such that $\rad_\F(t) \le \Psi(t)$ for every $t \in \N$ and argue that $\fh_\F$, the fractional Helly number of $\F$, is bounded. 

  \medskip

  By Levi's inequality~\eqref{eq:gradedhr} we have $\h_\F(t) \le \rad_\F(t) -1 \le \Psi(t) -1$. It follows that $\h_\F(t) < t$ for every $t \ge t_0$, so by Lemma~\ref{l:hellygrowth} we have $\h_\F \le t_0$. Hence, the graded version~\eqref{eq:colorful-helly-g} of the Holmsen-Lee inequality applies with $c=t_0$ and every $t$, that is $\ch^{(t_0)}_\F(t) \leq \max\pth{\Xi\pth{\Psi(t)},t_0}$.

  \medskip

  Let us apply Lemma~\ref{l:holmsen-graded} with $c=t_0$ and $\ell = \frac{t_1}{t_0}$. Observe that $\ell > c$ holds because $t_1 > {t_0}^2$, and that $\ch^{(c)}_\F(c\ell) \le \ell$ follows from the assumption that $\Xi\pth{\Psi(t_1)} \le \frac{t_1}{t_0}$, as
  \[  \ch^{(c)}_\F(c\ell) = \ch^{(t_0)}_\F\pth{t_1} \le \max\pth{\Xi\pth{\Psi(t_1)},t_0} \le \frac{t_1}{t_0} = \ell.\]
Hence, $\fh^{(t_0)}_\F \leq \ch^{(t_0)}_\F(t_1)$. Since $\h_\F \le t_0$ we have $\fh^{(t_0)}_\F = \fh_\F$ and the statement follows.  
\end{proof}

\appendix

\section{Homological minors of manifolds and their triangulations}\label{app:simpapp}

Here we give the proof of:

\setcounter{replemma}{7}
\begin{replemma}
    A simplicial complex $K$ is a homological minor of a compact PL manifold (possibly with boundary) $\M$ if and only if $K$ is a homological minor of some triangulation of~$\M$.
\end{replemma}
\begin{proof}
  Let $K$ and $L$ be simplicial complexes. A map $\tau':|K|\to |L|$ is \emph{an approximation} of a map $\tau:|K|\to |L|$ with respect to $L$ if for every $x\in |K|$, $\tau'(x)$ belongs to the inclusion-minimal closed simplex of $L$ containing the point $\tau(x)$. (See~\cite{zeeman1964relative}.) The relative simplicial approximation theorem~\cite{zeeman1964relative} states that any continuous map into $|L|$ whose restriction to the boundary is simplicial admits an approximation (with respect to $L$) that is a simplicial map leaving the boundary unchanged. Moreover, it gives a relative homotopy (fixing the boundary) between the continuous map and the simplicial approximation, which remains an approximation (with respect to $L$) for every $t\in[0,1]$.

\medskip
  
Now, let $a:C_*(K)\to\Csing(|T|)$ be a homological almost-embedding and let $T$ be a triangulation of $\M$ fine enough that no closed simplex of $T$ intersects the supports $\supp(a(\sigma))$ and $\supp(a(\tau))$ of two non-adjacent faces $\sigma,\tau$ of $K$. We apply the relative simplicial approximation theorem to construct a simplicial chain map $b \colon C_*(K)\longrightarrow C_*(T)$ satisfying the following property: for every $\sigma\in K$, denoting $a(\sigma)=\tau_0+\tau_1+\cdots+\tau_s$, we have $b(\sigma)=\tau_0'+\tau_1'+\cdots+\tau_s'$, where each $\tau_i'$ is a simplicial approximation of the singular simplex $\tau_i$. Our choice of $T$ then guarantees that $b$ is again a homological almost-embedding.

\medskip

Observe that the chain map $a$ is supported on a finite number of singular simplices $\tau:\Delta^i\to \M$. We inductively construct simplicial approximations of all such map, starting from $i=0$ to $i=\dim(K)$. In the process, we make sure that the approximation $\tau'$ of each singular simplex $\tau$ is compatible with the other approximations in the following sense: if $\partial \tau=\omega_0+\omega_1+\ldots+\omega_t$, we require the approximation $\tau'$ to satisfy $\partial\circ \tau'= \omega_0'+\omega_1'+\ldots+\omega_t'$, where the $\omega_j'$ are the previously chosen approximations.

\medskip

Specifically, suppose that we are processing a map $\tau:\Delta^t\to \M$. Let $\omega_0,\omega_1, \ldots, \omega_t$ denote the facets of $\tau$. By induction, we assume that we already computed a suitable simplicial approximation $\omega_i'$ of $\omega_i$ for $0 \le i \le t$. We cannot apply the relative simplicial approximation theorem to $\tau$ directly, as it may not be simplicial on the boundary $\partial \Delta^t$. We therefore define an auxiliary map $\widetilde \tau:\Delta^t\to\M$ with boundary is $\omega_0'+ \omega_1' + \ldots+\omega_t'$ by concatenation of the map $\tau$ and the homotopy between $\tau\circ\partial$ and $\tau'\circ\partial$ given by the previous use of simplicial approximation theorem to each map constituting $\partial\tau$. We then apply the relative simplicial approximation theorem to  $\widetilde \tau$. The simplicial map we obtain satisfies our compatibility condition.

\medskip

So suppose that we have our collection of simplicial approximations of all singular simplices in the support of the chain map $a$, where the simplicial approximations of $\tau$ being denoted $\tau'$. For every $\sigma\in K$, letting $a(\sigma)=\tau_1+\cdots+\tau_s$, we put $b(\sigma)=\tau_1'+\cdots+\tau_s'$. It is a chain map by the compatibility of our simplicial approximations, it is nontrivial because $a$ is, and the images of independent simplices have disjoint supports by our choice of $T$ and the property of simplicial approximations. Altogether, $b$ is a homological almost-embedding of $K$ into $T$.
\end{proof}

\section{Theorem~\ref{t:forbidden} for relaxed connectivity assumptions}\label{app:relax}

In this appendix we relax the conditions in Theorem~\ref{t:forbidden} that the $d$-manifold be $(\lceil d/2 \rceil -1)$-connected. 

\medskip

A classical method to increase connectivity of a manifold is to use surgeries. Let $\M$ a $d$-manifold and let $p \le d$. A manifold~$\M'$ is \defn{produced by a $p$-surgery} from $\M$ if there exists an embedding $\phi:\mathbb{S}^p\times \mathbb{D}^{q}\to \M$, with $p+q=d$, such that $\M'= \pth{M\backslash \mathring{\pth{\im(\phi)}}}\cup_{f_{|\mathbb{S}^p\times\mathbb{S}^{q-1}}}\pth{\mathbb{D}^{p+1}\times\mathbb{S}^{q-1}}$. It turns out that surgeries do not destroy homological minors of sufficiently small dimension.

\begin{lemma}\label{l:surgery-stability-minors}
  Let $K$ be a $k$-dimensional homological minor of a $2k$-manifold (possibly with boundary) $\M$. If $\M'$ is a manifold produced by a $p$-surgery from $\M$,
  with $p < k$, then $K$ is a homological minor of~$\M'$.
\end{lemma}
\begin{proof}
  Let $a:C_*(K)\to \Csing(\M)$ be a homological almost-embedding from $K$ to $\M$. We can assume that $A:=\bigcup_{\omega\in K}\supp(a(\omega))$ is a simplicial subcomplex inside some triangulation of $\M$. Let $\M'$ be a manifold produced by a $p$-surgery from $\M$. That is, there exists an embedding $\phi:\mathbb{S}^p\times\mathbb{D}^q\to \M$, with  $q = 2k-p$, such that $\M'= \pth{M\backslash \mathring{\pth{\im(\phi)}}}\cup_{f_{|\mathbb{S}^p\times\mathbb{S}^{q-1}}}\pth{\mathbb{D}^{p+1}\times\mathbb{S}^{q-1}}$. We define $\phi_0:\mathbb{S}^p\to\M$ by $\phi_0(x)=\phi(x,0)$ (which we can assume PL) and by the (PL) general position theorem for embeddings~\cite[$\S$5.3]{Rourke1972IntroductionTP}, there is a homeomorphism $h:\M\to\M$ (a perturbation of the identity) such that $dim\pth{h( A)\cap im(\phi_0)}\leq p+k-2k<0$. The chain map $b=h_{\#}\circ a$ is a homological almost-embedding from $K$ to $\M\setminus im(\phi_0)$. Since $\M\setminus im(\phi_0)$ embeds into $\M'$ via the inclusion $i$, the same chain map the chain map $i_\#\circ b$ shows that $K$ is a homological minor of $\M'$.
\end{proof}

A manifold is called \emph{stably parallelizable} or a \emph{$\pi$-manifold} if its tangent bundle is stably parallelizable, meaning that its sum with some trivial bundle is trivial (see~\cite{milnor}). Milnor proved~\cite[Corollary to Theorem~$2$]{milnor} that any compact triangulated differentiable $\pi$-manifold of dimension $2k$ can be made $(k-1)$-connected by a finite sequence of $p$-surgeries with $p<k$. The same result holds~\cite[Corollary to Theorem~$3$]{milnor} for any compact triangulated differentiable manifold that is $(k-1)$-parallelizable (i.e, if its tangent bundle restricted to its $(k-1)$-skeleton is trivial). Lemma~\ref{l:surgery-stability-minors} ensure that these surgeries preserve $k$-dimensional homological minors, so we get the following generalization of Theorem~\ref{t:forbidden}:

\begin{theorem}\label{t:pi-manifold}
  Let $\M$ be a compact PL manifold (possibly with boundary) of dimension $d$. If $\M$ is a $\pi$-manifold or is $(k-1)$-parallelizable, then there exists $N$ such that $\Delta_N^{\pth{\lceil \frac{d}2\rceil}}$ is not a homological minor of $\M$.
\end{theorem}

\noindent
For example, $\mathbb{T}^d$ is parallelizable for every $d \ge 1$, so for $d \ge 3$ the manifold $\mathbb{T}^d$ is covered by Theorem~\ref{t:pi-manifold} (but not by Theorem~\ref{t:forbidden}). Also, any closed orientable $3$-manifold is parallelizable, and any product of stably parallelizable manifolds is stably parallelizable. Examples of manifolds that are covered by neither Theorem~\ref{t:forbidden} nor Theorem~\ref{t:pi-manifold} is $\mathbb{RP}^d$ for even values of $d$.

\section{Convexity spaces VS set systems}\label{app:convspaces}

In this appendix, we provide for completeness a discussion of why the results of Levi~\cite{Levi_radon_implies_helly} and Holmsen and Lee~\cite{Holmsen2021} stated for convexity spaces hold in the context of set systems. See also the discussion in the v3 of Patakova's preprint~\cite[Appendix~A]{Pat_kov__2024}.

\bigskip

A common setting for studying generalizations of convexity is to consider a set $X$ and a set $\C$ of subsets of $X$ satisfying three properties:

\begin{description}
\item[(i)] $\emptyset$ and $X$ are in $\C$,
\item[(ii)] $\C$ is closed under intersections, and
\item[(iii)] $\C$ is closed under unions of nested families.
\end{description}

\noindent
A pair $(X,\C)$ satisfying all three properties is called a \defn{convexity space}. These structures were extensively investigated by associating to a general convexity space various parameters such as its Helly number, its Radon number, etc. and establishing relations between these parameters that hold for arbitrary convexity spaces. This includes for instance Levi's inequality between Helly and Radon numbers~\cite{Levi_radon_implies_helly}.

\medskip

Let $\F$ be a set system with ground set $X$. The \defn{$\F$-convex hull} $\conv_{\F}(P)$ of a subset $P\subset X$ is the intersection of all the members of $\F$ that contain $P$. When $\F$ is not a convexity space, we do not have that $\conv_{\F}(P)$ is in $\F$ but this is of no concern to us.

\begin{itemize}
\item The \defn{Helly number} $\h_\F$ of $\F$ is the smallest integer $h$ with the following property: If in a finite subfamily $\G \subset\F$, every $h$ members of $\G$ intersect, then $\G$ has nonempty intersection. If no such $h$ exists, we set $\h_\F=\infty$.\medskip

\item The \defn{Radon number} $\rad_\F$ of $\F$ is the smallest integer $r$ such that every $r$-element subset $S\subset X$ can be partitioned into two nonempty parts $S=P_1\sqcup P_2$ such that $\conv_{\F}(P_1) \cap \conv_{\F}(P_2) \neq \emptyset$. If no such $r$ exists, we set $\rad_\F=\infty$.
\end{itemize}

\noindent
Note that both definitions make sense also without axioms~(i)--(iii) of convexity spaces. Here is a proof of Levi's inequality:

\begin{lemma}
  For every set system $\F$ with finite Radon number, $\h_\F \le \rad_\F-1$.
\end{lemma}
\begin{proof}
  Let $\F$ be a set system and let $X$ denote its ground set. Let $\h(\F)$ denote the set of all subsets $W \subseteq \F$ with the following properties: $W$ has empty intersection and every proper subset of $W$ has nonempty intersection. Let $|\cdot|$ denotes the cardinality. Note that:

  \begin{itemize}
  \item Every $W \in \h(\F)$ refutes the possibility that the Helly number of $\F$ be at most $|W|-1$.\medskip
    
  \item Conversely, these are the inclusion-minimal obstructions: For any $k \in \N$ and finite subset $\G \subset\F$, if every $k$ members of $\G$ intersect and $\G$ has empty intersection, then there exists $W \subseteq \G$ such that $W \in \h(\F)$ and $|W| \ge k+1$.\medskip
  \end{itemize}

  \noindent
  It follows that $\h_\F$ is finite if and only if $\h(\F)$ is finite, and $\h_\F = \max_{W \in \h(\F)} |W| - 1$.
  
  \medskip

  Let $W \in \h(\F)$. For every $A \in W$ let $x_A \in \cap_{B \in W\setminus\{A\}} B$. Let $S = \{x_A \colon A \in W\}$. Consider an arbitrary partition $S=P_1\sqcup P_2$. 
  For $i=1,2$ define $W_i = \{A \in W \colon x_A \notin P_i\}$. 

  Observe that $W = W_1 \cup W_2$ and that $W_i$ contains $\{A \in W \colon P_i \subset A\}$. It follows that

  \[\begin{aligned}
  \conv_{\F}(P_1) \cap \conv_{\F}(P_2) & = \pth{\bigcap_{\stackrel{A \in \F}{P_1 \subset A}}A} \cap \pth{\bigcap_{\stackrel{A \in \F}{P_2 \subset A}}A}\\
  & \subset \underbrace{\pth{\bigcap_{\stackrel{A \in W}{P_1 \subset A}}A}}_{\subset\bigcap_{A \in W_1}A} \cap \underbrace{\pth{\bigcap_{\stackrel{A \in W}{P_2 \subset A}}A}}_{\subset\bigcap_{A \in W_2}A} \subset \pth{\bigcap_{A \in W}A} = \emptyset
  \end{aligned}\]

  \noindent
  Thus, there is no partition $S=P_1\sqcup P_2$ such that $\conv_{\F}(P_1) \cap \conv_{\F}(P_2) \neq \emptyset$. It follows that $\rad_\F \ge |W|$. Since this holds for every $W \in \h(\F)$, we have $\rad_\F \ge \max_{W \in \h(\F)}|W| = \h_\F+1$.
\end{proof}

The proof of our Corollary~\ref{c:itgrows!} uses Theorem~1.1 of Holmsen and Lee~\cite{Holmsen2021}, which is stated for convexity spaces. The proof of~\cite[Theorem~1.1]{Holmsen2021} uses two ingredients:

\begin{enumerate}
\item a colorful Helly theorem for convexity spaces with bounded Radon number~\cite[Theorem~2.2]{Holmsen2021}, where convexity spaces are only used to apply Levi's inequality and an inequality of Jamison~\cite[Ineq.~(6)]{jamison1981partition} on partition numbers, and\medskip
  
\item a lower bound on the size of the largest clique in a uniform hypergraph with a forbidden pattern~\cite[Theorem~1.2]{Holmsen2020}.
\end{enumerate}

\noindent
This only leaves us with Jamison's inequality to examine. Jamison's inequality is about analogues of Radon numbers for more than $2$ parts, but dealing with multisets.

\medskip

Let $\F$ be a set system with ground set $X$. Given a multiset $S$ of $X$ we write $\bar S$ for the set of elements that appear in $S$, and define $\conv_\F(S) := \conv_\F(\bar S)$. A \defn{partition} of a multiset $S$ into $k$ parts is a union of $k$ multisets $(S_1, S_2, \ldots, S_k)$ such that $\bar S = \bar S_1 \cup \bar S_2 \cup \ldots \cup \bar S_k$ and the multiplicity of an element in $S$ equals the sum of its multiplicities in $S_1, S_2, \ldots, S_k$. 

\begin{itemize}
\item For $k \ge 2$, the \defn{$k$th partition number} $\rad^{(k)}_\F$ of $\F$ is the smallest integer $r$ such that every multiset $S$ of $X$ of size $r$ can be partitioned into $k$ nonempty parts $(P_1, P_2, \ldots, P_k)$ such that $\conv_{\F}(P_1) \cap \conv_{\F}(P_2) \cap \ldots  \cap \conv_{\F}(P_k)\neq \emptyset$. If no such $r$ exists, we set $\rad_\F^{(k)}=\infty$.
\end{itemize}

\noindent
In particular $\rad^{(2)}_\F = \rad_\F$. Jamison's inequality asserts that when $\F$ is a convexity space, $\rad^{(k)}_\F \le {\rad^{(k)}_\F}^{\log_2 k}$. It follows from a recursive formula~\cite[Theorem~1]{jamison1981partition} whose proof uses no property of convexity spaces. For completeness, we reproduce here the proof of a slightly simpler version of that inequality:

\begin{lemma}
  For any set system $\F$ and any $m,n \in \N$ we have $\rad^{(mn)}_\F \le \rad^{(m)}_\F\rad^{(n)}_\F$.
\end{lemma}
\begin{proof}
  Let $X$ be the ground set of $\F$ and let $S$ be a multiset of size $\rad^{(m)}_\F\rad^{(n)}_\F$ of $X$. Fix any partition of $S$ into $k = \rad^{(m)}_\F$ multisets $(S_1, S_2, \ldots, S_k)$, each of size $\rad^{(n)}_\F$. Each $S_i$ admits a partition $(P^i_1, P^i_2, \ldots, P^i_n)$ such that $\conv_{\F}(P^i_1) \cap \conv_{\F}(P^i_2) \cap \ldots  \cap \conv_{\F}(P^i_n)\neq \emptyset$. Let $p_i$ be an element in that intersection. Now let $P$ denote the multiset $\{p_1, p_2, \ldots, p_k\}$. Since $k = \rad^{(m)}_\F$ there is a partition of $P$ into $m$ multisets $(P_1,P_2, \ldots, P_m)$ and an element $x \in X$ such that $x \in \conv_{\F}(P_1) \cap \conv_{\F}(P_2) \cap \ldots  \cap \conv_{\F}(P_m)$. Now, for every $a \im [m]$ let $A(a) = \{i \colon p_i \in P_a\}$. For any $(a,b) \in [m]\times [n]$ we put $Q_{a,b} := \cup_{i \in A(a)} P^i_b$. The $Q_{a,b}$ therefore form a partition of $S$ into $mn$ multisets. Moreover, since for all $i$ and $b$, $p_i\in \conv_{\F}(P_b^i)$, we have that each $\conv_\F(Q_{a,b})$ contains every $p_i$ with $i \in A(a)$, and therefore also contains $x$.
\end{proof}

\bigskip

Finally, we note that the proofs by Alon et al.~\cite{transversal-hypergraph} that a Fractional Helly theorem implies the existence of weak $\varepsilon$-nets, $(p,q)$-theorems, etc. does require the set system to be closed under intersection.

\section{An upper bound on the growth of graded Radon numbers}\label{app:rab}

A pigeonhole argument yields that for every set system $\F$ the graded Radon numbers satisfy $\rad_\F(t)< 2^t$. Indeed, the (graded) Radon numbers do not change when we identify two elements in the ground set of $\F$ that belong to exactly the same members of $\F$. It turns out that a much sharper bound holds.

\begin{proposition}
 For any set system $\F$ and any $t \in \N$, we have $\rad_{\mathcal{F}}(t)\leq t+1$.
\end{proposition}
\begin{proof}
  Suppose there exists $\mathcal{F}' \subseteq \F$ of size $t$ of Radon number greater than $t+1$, that is with some set $S=\{p_1,\dots,p_{t+1}\}$ of $t+1$ points in $X$ with no $\F'$-Radon partition.

  \medskip

  Observe that if a partition $\mathcal{P}=(\mathcal{P}_0,\mathcal{P}_1)$ of $\{p_1,\dots,p_{t+1}\}$ is not a Radon partition, then the $\F'$-hull of $\mathcal{P}_0$ does not contain any member of $\mathcal{P}_1$. In particular, there exists some $F_i\in\F'$ that contains $\mathcal{P}_0$ but not $\mathcal{P}_1$.

  \medskip

  Let us apply this observation to the partitions $\mathcal{P}^{(j)}=(\{p_1,\dots,p_{t+1}\}\backslash \{p_j\},\{p_j\})$ for $j\in\{1,\dots,t\}$. For every $j \in [t]$ there exists $F_{i_j}\in\mathcal{F}'$ such that $F_{i_j} \cap S = S \setminus \{p_j\}$. Note that the $\{F_{i_j}\}_{j \in [t]}$ must be pairwise disjoint, that is $\F' = \{F_{i_j}\}_{j \in [t]}$.

  \medskip

  To conclude, observe that no member of $\F'$ contains $\{p_1,p_2, \ldots, p_t\}$. We thus have
  \[ \conv_{\mathcal{F}'}(\{p_1,\dots,p_t\})=X,\]
  which contains $p_{t+1}$. Altogether, $(\{p_1,\dots,p_t\},\{p_{t+1}\})$ is a Radon partition, which contradicts the initial assumption.
\end{proof}

\end{document}